\documentclass[conference,a4paper]{IEEEtran}
\IEEEoverridecommandlockouts
\usepackage[T1]{fontenc}
\usepackage[utf8]{inputenc}
\usepackage{microtype}
\usepackage{cite}
\usepackage{amsmath,amssymb,amsfonts}
\usepackage{amsthm}
\usepackage{graphicx}
\usepackage{xcolor}
\usepackage{hyperref}
\usepackage{orcidlink}
\usepackage{balance}

\definecolor{PaperLinkBlue}{HTML}{0057FF}
\definecolor{NeonRefGreen}{HTML}{39FF14}
\hypersetup{
  colorlinks=true,
  linkcolor=PaperLinkBlue,
  urlcolor=PaperLinkBlue,
  citecolor=NeonRefGreen,
  pdfauthor={Markus Baumann, Michael Poppel, Thomas Gabor, Maximilian Zorn, Claudia Linnhoff-Popien, and Jonas Stein},
  pdftitle={Exploiting Symmetry in Quantum Reservoir Computing},
  pdfsubject={A concise architecture and evaluation of symmetry-aware quantum reservoir computing}
}

\graphicspath{{gfx/}}
\newcommand{\Tr}{\operatorname{Tr}}
\newtheorem{proposition}{Proposition}
\newtheorem{lemma}{Lemma}

\begin{document}

\title{Exploiting Symmetry in Quantum Reservoir Computing}

\author{%
\IEEEauthorblockN{%
Markus Baumann\IEEEauthorrefmark{1}\orcidlink{0009-0007-3575-1006}\thanks{Correspondence should be addressed to \href{mailto:markus.baumann@campus.lmu.de}{markus.baumann@campus.lmu.de}.},
Michael Poppel\IEEEauthorrefmark{1}\orcidlink{0009-0005-1141-0974},
Thomas Gabor\IEEEauthorrefmark{2}\orcidlink{0000-0003-2048-8667},\\[-0.2ex]
Maximilian Zorn\IEEEauthorrefmark{1}\orcidlink{0009-0006-2750-7495},
Claudia Linnhoff-Popien\IEEEauthorrefmark{1}\orcidlink{0000-0001-6284-9286},
and Jonas Stein\IEEEauthorrefmark{1}\orcidlink{0000-0001-5727-9151}
}
\IEEEauthorblockA{\IEEEauthorrefmark{1}\textit{QAR-Lab, Department of Computer Science, LMU Munich, Munich, Germany}}
\IEEEauthorblockA{\IEEEauthorrefmark{2}\textit{Department of Computer Science, University of Exeter, Exeter, United Kingdom}}
}
\maketitle
\thispagestyle{empty}
\pagestyle{empty}

\begin{abstract}
Quantum reservoir computing (QRC) uses a quantum processor without training
it. The input is encoded into a quantum state, a fixed random circuit evolves
it, selected observables are measured, and only a simple linear readout is
trained on the measured values. We study QRC for forecasting on a ring of
sensors, such as weather stations around a circle of latitude. On such a ring,
the same physical rule governs every position, so a model that respects this
symmetry can learn one shared local rule from all sensors at once instead of a
separate rule per sensor. This is especially valuable when data is scarce.
Because the readout sees only the measured numbers, the symmetry can be lost
even when the quantum state respects it. We show how to preserve the symmetry by measuring every observable
together with all its shifted copies and by using one shared rule for encoding,
circuit, and readout, and we prove that this construction is sufficient. In a
controlled audit the aligned design consistently outperforms misaligned
alternatives, and the advantage persists in noisy simulations, on IBM hardware,
and on real weather data.

\end{abstract}


\section{Introduction}
\label{sec:introduction}
A common way to use near-term quantum hardware for learning is to keep the
quantum part fixed and train only a classical model on its outputs
\cite{preskill2018quantum,biamonte2017quantum}. Quantum reservoir computing
(QRC) is such a scheme. The input is encoded into a quantum state, a fixed
circuit (the reservoir) transforms it, chosen observables are measured, and
a linear readout is trained on the measured values
\cite{fujii2017harnessing,ghosh2019quantum,nakajima2019boosting}. Skipping
variational circuit training is attractive, but it has a sharp consequence.
The trained part of the model never sees the quantum state, only the
finitely many numbers that were measured.

\begin{figure}[!t]
    \centering
    \includegraphics[width=\columnwidth]{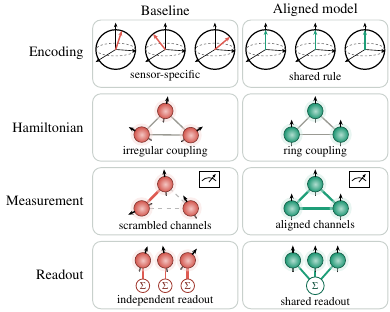}
    \caption{The four stages of a QRC model where the ring symmetry can be
    kept or broken. In the aligned model (right), every sensor follows the
    same rule: a shared input encoding, identical couplings around the ring,
    matching observables at every position, and one readout for all sensors.
    In the baseline (left), each stage is instead sensor-specific.}
    \label{fig:pipeline}
\end{figure}

We study this constraint for fields on a ring of $K$ sensors, such as
stations around a circle of latitude. On a homogeneous ring, if the input
pattern moves by $s$ sensors, the correct forecast moves by $s$ sensors as
well. Writing $S_s$ for this shift, the requirement reads
$F(S_sx)=S_sF(x)$. Such a model is called \emph{equivariant}. Equivariance
is valuable on small data because whatever the model learns at one position
applies at every other. In variational quantum learning, symmetry is built
into trainable embeddings, circuits, or measurements
\cite{meyer2022exploiting,schatzki2024theoretical,nguyen2022theory}. In QRC
only the readout is trained, so the symmetry must also survive the
measurement step that connects the reservoir to that readout.

Figure~\ref{fig:pipeline} shows the four places where the shared cyclic
rule may be kept or broken across encoding (E), reservoir dynamics (H), measured
observables (O), and readout (R). Our claim is about this architecture, not
all of QRC. Keeping the rule at all four stages is provably sufficient for
equivariance, and in our experiments the stages help mostly in combination.
Other architectures may achieve the same protection differently.

\emph{Contributions.}
\begin{itemize}
\item We identify the measurement/readout interface as a place where QRC
loses symmetry, and we separate the two ways this happens: shifted copies
of an informative observable are never measured, or they are measured but
attached to the wrong sensors.
\item We show how to build the pipeline so that neither failure occurs:
when encoding, dynamics, measurements, and readout all follow the same
shift symmetry, the model is provably shift-equivariant.
\item We test all $2^4$ on/off combinations of these four stages, plus
matched controls, on three synthetic systems, a dissipative reservoir,
finite-shot IBM hardware, and a limited WeatherBench~2 experiment
\cite{rasp2024weatherbench2,hersbach2020era5}. The gap between the fully
aligned and the fully untied model shows the size of the combined effect.
\end{itemize}

\section{Related Work and Positioning}
\label{sec:related}
Reservoir computing trains a linear readout on top of a fixed nonlinear
dynamical system \cite{maass2002real,jaeger2004harnessing,lukosevicius2009reservoir}.
QRC keeps this split but replaces the classical reservoir state with measured
expectation values of a quantum system
\cite{fujii2017harnessing,ghosh2019quantum,nakajima2019boosting}. Recent QRC
work studies reservoir design, repeated input injection, coherence,
measurement selection, and temporal processing
\cite{mujal2023timeseries,ahmed2024prediction,hu2024coherence}. This work
mostly treats measurements as a source of rich features. We study their
limiting role instead: since only the readout is trained, everything the
model learns about the quantum system must pass through the measured
values.

Building a known symmetry into a model is a standard way to learn more from
limited data, first in classical learning
\cite{lecun1998gradient,cohen2016group,bronstein2017geometric} and now in
variational quantum models through equivariant embeddings, circuits,
kernels, and measurements
\cite{meyer2022exploiting,schatzki2024theoretical,nguyen2022theory}. All of
these methods shape the parts of the model that are trained, so they do not
directly transfer to QRC, where the quantum side is fixed.

Measuring an observable together with all its shifted copies, which we call
observable-orbit completion, is in spirit a standard group-orbit
construction rather than a new group-theoretic result. The completed
channels span the smallest space of observables containing the seeds and
closed under shifts. Barbosa et al.\
align a classical reservoir with task symmetry \cite{barbosa2021symmetry}.
The difference in QRC is that a classical readout can inspect the whole
reservoir state, while a quantum readout sees only the selected
measurements, so symmetry can be lost at an interface classical reservoirs
lack in this form. Sannia et al.\ show that symmetry can reduce feature
concentration in QRC \cite{sannia2025concentration}. We ask a
complementary question: are the shifted copies of the relevant observables
measured at all, and are they attached to the sensors in a way a shared
readout can use?

Three things are new here: the diagnosis of the two failure modes at the
measurement/readout interface, an architecture whose four stages can be
switched on and off independently, and matched validation across
simulation, hardware, and real weather data.

\section{Symmetry at the Measured Interface}
\label{sec:interface}
\begin{figure*}[!t]
    \centering
    \includegraphics[width=\textwidth]{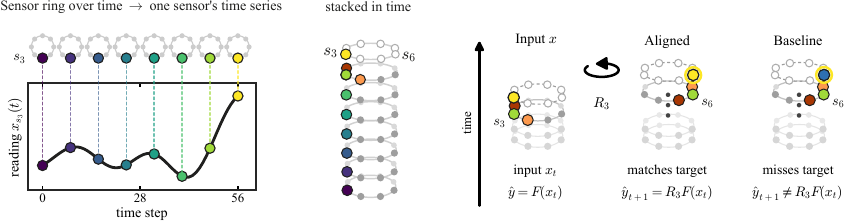}
    \caption{Sensing geometry and the effect of alignment on representative
    simulated $K=8$ advection data. The left panel shows snapshots of the sensor ring and the
    resulting time series at sensor $s_3$. The middle panel shows the same snapshots stacked
    in time. The right panel shows that when a local pattern reappears shifted by three sensors,
    the desired prediction shifts by the same amount.}
    \label{fig:geometry}
\end{figure*}

The model receives one snapshot $x_t\in\mathbb{R}^K$, one number per
sensor, and predicts each sensor's change in the next step. It outputs the
residual $r_t=x_{t+1}-x_t$ and returns $\hat x_{t+1}=x_t+\hat r(x_t)$. With
sensor indices modulo $K$, let $(S_sx)_i=x_{i-s}$ denote the field shifted
by $s$ positions. Because the ring is homogeneous, the prediction rule
should obey
\begin{equation}
    \hat r(S_sx)=S_s\hat r(x).
    \label{eq:desired-equivariance}
\end{equation}
Figure~\ref{fig:geometry} gives the picture. If a local pattern reappears a
few sensors further along the ring, its prediction should reappear there
too. This property is called equivariance. Its practical value is data
efficiency. Whatever the model learns at one sensor applies at every other,
so all $K$ sensors contribute to learning one shared rule. This only works
if every sensor's features mean the same thing. Otherwise pooling them
mixes incomparable numbers.

After encoding and the fixed reservoir evolution, write the quantum state
as $\rho(x)$. Each sensor $i$ is assigned a block of measured features
\begin{equation}
    \phi_i(x)_a=\Tr[O_{i,a}\rho(x)],
    \label{eq:local-feature}
\end{equation}
the expected outcomes of measuring $O_{i,a}$ in the state $\rho(x)$. A
shared local readout predicts $\hat r_i(x)=w^\top\phi_i(x)+b$,
with the same weights $w$ and bias $b$ at every sensor. The readout never
sees $\rho(x)$. It sees only these finitely many numbers and, being linear,
can output nothing but weighted sums of them. Two distinct failures follow.

\emph{Missing shifted observables.} Suppose a local signal shows up in
$\langle Z_0\rangle$ and the measurement list contains $Z_0$ but not its
shifted copy $Z_1$. When the input moves by one sensor, the signal moves to
$\langle Z_1\rangle$, which is never recorded. In general, a shifted copy
of a relevant quantity can lie outside the span of the recorded features,
and no linear readout can reconstruct a number that was never recorded.
Task-specific correlations may partly compensate, but nothing guarantees
it. We call this the \emph{measured-span obstruction}.

\emph{Measured but wrongly attached.} In the second failure mode,
everything relevant is measured, but the values are attached to the wrong
sensors. All the information is still there, and a large model that reads
all sensors at once could still find it. A shared rule cannot, because it
relies on each sensor's values meaning the same thing, and after the
mix-up they do not. The model does not lose information here, it loses the
ability to learn one rule from all sensors together. The dense-readout and
orthogonal-recombination controls isolate this case.

The repair fixes both problems in one step. Decide on the observables
once, at one reference sensor, and let every sensor measure its own
shifted copy of each of them. For a seed $O$, sensor $i$ measures
\begin{equation}
    O_i=T_iOT_i^\dagger,
    \qquad i=0,\ldots,K-1,
    \label{eq:observable-orbit}
\end{equation}
where the unitary $T_i$ moves the qubits the same way the shift moves the
sensors. We call this \emph{observable-orbit completion}, because the
shifted copies of a seed form its orbit under the shift group. Now all
sensors measure the same quantities in matching order, so nothing is
missing and nothing sits in the wrong place.

\begin{proposition}[The aligned model is shift-equivariant]
\label{prop:shared-readout}
Suppose the feature blocks move with the input: after a shift by $s$, each
sensor receives the block that the sensor $s$ positions back received
before, $\phi_i(S_sx)=\phi_{i-s}(x)$. If the same readout
$\hat r_i(x)=w^\top\phi_i(x)+b$ is used at every sensor, then
$\hat r(S_sx)=S_s\hat r(x)$, and the one-step forecast $x+\hat r(x)$ shifts
with the input as well.
\end{proposition}

\begin{proof}
After the shift, each sensor sees exactly the numbers that the sensor $s$
positions back saw before, and it applies the same weights and bias to
them. It therefore outputs the prediction that sensor made before. Every
prediction moves by $s$ positions, which is the claim.
\end{proof}

\begin{lemma}[The aligned pipeline shifts the features correctly]
\label{lem:feature-shift}
Let $T_s$ be the qubit relabeling induced by the input shift $S_s$. Assume
the reservoir state shifts along,
\begin{equation}
    \rho(S_sx)=T_s\rho(x)T_s^\dagger,
    \label{eq:state-covariance}
\end{equation}
and that the measured observables come in the complete orbits of
Eq.~\eqref{eq:observable-orbit}. Then the feature blocks move with the
input, $\phi_i(S_sx)=\phi_{i-s}(x)$.
\end{lemma}

\begin{proof}
Consider one measured number, the expectation of the copy $O_i$ of some
seed. Its value in the shifted state is
\begin{equation*}
\Tr[O_i\,T_s\rho(x)T_s^\dagger]
=\Tr[(T_s^\dagger O_iT_s)\,\rho(x)]
=\Tr[O_{i-s}\,\rho(x)],
\end{equation*}
where the first step moves $T_s$ around the trace and the second uses that
relabeling the copy at sensor $i$ by $T_s^\dagger$ gives the copy at sensor
$i-s$. So every number measured after the shift was already measured before
the shift, $s$ sensors back, which is the claimed relocation of blocks.
\end{proof}

Lemma~\ref{lem:feature-shift} says the aligned quantum pipeline delivers
features that move correctly with the input.
Proposition~\ref{prop:shared-readout} turns that into an equivariant
prediction. The audit below tests whether this easily stated condition
identifies a design bottleneck that matters in practice.

\section{Architecture and Audit}
\label{sec:method}
\subsection{Architecture and four switches}
Every model predicts a one-step residual on a cyclic field using one qubit
per sensor. Sensor $i$ is encoded from the five-site spatial patch
\begin{equation}
p_{t,i}=(x_{t,i-2},x_{t,i-1},x_{t,i},x_{t,i+1},x_{t,i+2}),
\end{equation}
with indices modulo $K$. This choice matches how the test systems behave:
the next change at a sensor is driven mainly by its close neighborhood, so
each qubit is shown its own sensor together with two neighbors on each
side. A fixed two-layer random reservoir circuit follows. The only trained
component is ridge regression, linear regression with a penalty on large
weights. The audit toggles four independently
defined switches, illustrated in Fig.~\ref{fig:pipeline}.

The circuit starts from the uniform product state $|+\rangle^{\otimes K}$.
Each layer turns the local patch into single-qubit rotation angles, then
applies a transverse-field-Ising-style block with local fields and
couplings to first and second neighbors. All angles, fields, and couplings
are drawn once per paired reservoir seed and never optimized. The
experiment therefore changes how much structure is shared across space, not
how much quantum training is done.

\emph{Encoding (E)} uses one patch-to-angle rule at every sensor when
active, and an independent rule per sensor otherwise. \emph{Dynamics (H)}
ties the local fields and nearest/next-nearest couplings around the ring
when active. The inactive version draws position-specific parameters. In
both cases the reservoir stays fixed and random. \emph{Observables (O)}
sets how the measured channels are assigned to sensors. When active, every
sensor measures the same observables as the reference sensor, shifted to
its own position. When inactive, just as many channels are measured in
total, but they are assigned to the sensors in one fixed scrambled order,
so no sensor measures the observables it would have under alignment.
Scrambles that happen to equal a cyclic shift are rejected. \emph{Readout
(R)} pools the data of all sensors and time steps into one shared ridge
model when active, and fits a separate model per sensor otherwise.

Each sensor receives 61 features: one constant bias plus 60 Pauli
expectation values, of which 15 involve a single qubit, 28 involve two
qubits, eight involve three, and nine are averages of an observable taken
around the whole ring. The O switch therefore changes only which sensor
gets which channels, not how many there are, how many qubits they touch,
or which Pauli operators appear. We evaluate all $2^4=16$ on/off
combinations, from the fully untied 0000 to the fully aligned 1111. All of
them measure equally many channels, so none wins by measuring more. Only
the number of trained parameters differs, and deliberately so: whether
sharing one model helps is exactly what the R switch tests.

\subsection{Digital audit and realistic controls}
The main audit uses exact statevector expectation values, so no shot noise
enters. It covers three periodic systems. These are constant-speed advection, where a
pattern simply travels around the ring, Lorenz-96 with forcing $F=8$
\cite{lorenz1996predictability}, and the spatiotemporally chaotic
Kuramoto-Sivashinsky equation
\cite{kuramoto1978diffusion,sivashinsky1977nonlinear}. Each system runs at
$K\in\{8,10,12\}$, giving nine settings. The data-scarce audit uses 100
training, 80 validation, and 160 test steps with 100 paired reservoir seeds
per setting. Preprocessing statistics come from training data only.

Paired means that within one cell all 16 configurations reuse the same
trajectory split, reservoir draw, and ridge grid, so each switch is
compared within-seed rather than across unrelated random models. Four
controls move beyond the extreme 0000 endpoint.

\begin{itemize}
\item Partial configurations, with only some of the four switches active,
represent realistic systems that already preserve part of the spatial
structure.
\item A dense readout lets every output sensor use all $K\times61$ measured
channels at once. If it performs well on scrambled features, the
information is still there and only its organization is wrong.
\item Rotation augmentation trains the fully untied model on all $K$
rotated copies of every training input. This tests whether showing the
symmetry through the data can replace building it into the model, at $K$
times the reservoir-evaluation cost.
\item With encoding, dynamics, and readout aligned, we compare aligned,
scrambled, and random choices of the measurements, from 4 to 61 channels
per sensor, always with the same channel count, locality, Pauli weight,
and alphabet. A further control, the orthogonal recombination, mixes the
measured values without losing any information and so changes only their
coordinates.
\end{itemize}

These controls separate the two failure modes introduced above. The fixed
scramble and the orthogonal recombination lose no information and only
attach it to the wrong coordinates, so any damage they cause to the shared
readout reflects the second failure mode, broken organization. A random
map, in contrast, can genuinely leave out shifted copies of informative
observables and so create the first, a true measured-span obstruction. We
use that term only for this case: a gap between the aligned and the
scrambled model does not by itself mean information was missing.

\begin{figure*}[!t]
    \centering
    \includegraphics[width=\textwidth]{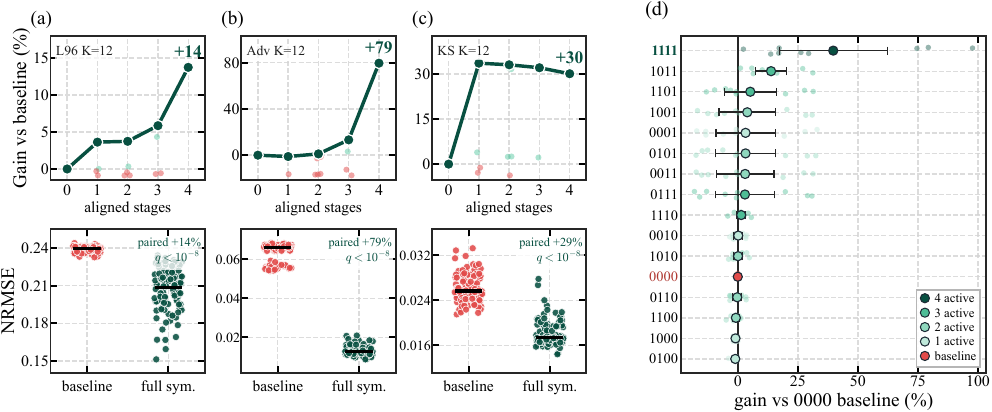}
    \caption{Paired fixed-budget E/H/O/R audit. All configurations use 61
    features per sensor. Panels (a) to (c) show the best gain at each number of
    aligned stages and paired 0000-versus-1111 test NRMSE for representative
    Lorenz-96, advection, and Kuramoto-Sivashinsky settings. Faint dots show
    individual configurations or seeds. Panel (d) ranks all 16 configurations
    by mean paired gain across nine settings. Error bars are 95\% confidence
    intervals. The $+40.1\%$ endpoint measures the combined full-pipeline
    effect. Partial and matched controls determine practical relevance.}
    \label{fig:audit}
\end{figure*}

\subsection{Metrics and statistical units}
For test targets $y_{t,i}$ and predictions $\hat y_{t,i}$, we report
\begin{equation}
\mathrm{NRMSE}=
\sqrt{\frac{\sum_{t,i}(\hat y_{t,i}-y_{t,i})^2}
{\sum_{t,i}(y_{t,i}-\bar y)^2}},
\end{equation}
the prediction error relative to always predicting the mean, and the paired
gain of model $m$ over baseline $b$,
\begin{equation}
\mathrm{gain}(m,b)=100\%
\frac{\mathrm{NRMSE}_b-\mathrm{NRMSE}_m}{\mathrm{NRMSE}_b},
\end{equation}
the percentage by which $m$ reduces the error of $b$. The ridge penalty is
selected on validation data from the same logarithmic grid $[10^{-8},10^4]$
for every model. Confidence intervals resample paired reservoir seeds.
Two-sided Wilcoxon signed-rank tests use Benjamini-Hochberg correction
within each experiment. A result is called robust-positive only when its
mean and 95\% bootstrap lower bound are positive, the adjusted $q<0.05$,
and at least 70\% of seed pairs improve.

\subsection{Support settings}
Three further settings probe robustness beyond exact digital simulation.
A physically different dissipative spin-ring reservoir
spans dephasing rates $\gamma\in\{0,0.03,0.08\}$ and exact, 1024-shot, and
4096-shot features. It covers the same nine settings at 400 training steps
and adds $K=16$ diagnostics. The IBM pilot uses $K=8$ advection, a 12/6/12
train/validation/test split, 10 paired reservoir seeds, five fixed
scrambled assignments, and 256 shots per circuit. Aligned and scrambled
models are built from the same raw counts, holding shot noise and
calibration drift fixed.

The real-data stress test uses WeatherBench~2 and ERA5
\cite{rasp2024weatherbench2,hersbach2020era5}. It uses 500-hPa geopotential
height sampled every six hours from January to May 2015, with a 24-hour
forecast horizon. Five latitude rings each carry eight equally spaced
longitude sensors. Training prefixes range from 10 to 80 steps.

\section{Results}
\label{sec:results}
\subsection{Full-pipeline magnitude}
Relative to the deliberately fully untied 0000 endpoint, the fully aligned
1111 model improves test NRMSE in all nine digital settings, with a mean
paired gain of $40.1\%$ (Fig.~\ref{fig:audit}). This is the full cross-stage
interaction within the factorial audit, not improvement over every QRC
architecture. The controls below establish practical relevance. Activating
E, H, O, or R alone gives $-0.4\%$, $-0.4\%$, $+0.6\%$, and $+3.8\%$, and
E+H gives $-0.1\%$. The strongest three-switch model, E+O+R, reaches
$+14.4\%$. E+H+O reaches only $+2.1\%$. Thus 1111 ranks first among all 16,
showing synergy in this construction rather than universal necessity of
these four switches. Gains are largest on advection, remain substantial on
Kuramoto-Sivashinsky, and are smaller but positive on Lorenz-96.

\subsection{Practical and mechanism controls}
A dense readout over all $K\times61$ channels gains $+11.9\%$, versus
$+39.8\%$ for alignment under the same refit. Dense fits on aligned and
scrambled features differ by at most $5\times10^{-8}$ NRMSE, confirming
that the scramble permutes rather than removes information. With only 100
training steps, however, the unstructured readout exploits it less
effectively. Rotation augmentation reaches $+37.9\%$, but costs $K$
reservoir evaluations per training input and remains between 5 and 10 points behind
alignment on every Lorenz-96 setting. Applied to misaligned shared features,
it reaches only $+7.9\%$.

In the equal-budget sweep, orbit-aligned measurements beat matched random
local maps for every task and budget from 4 to 61 features per sensor, and all 45
combinations of task, size, and budget are robust-positive. A span-preserving orthogonal
recombination leaves dense-readout performance unchanged but cuts the shared
gain from $+39.8\%$ to $+2.3\%$. Thus missing orbit elements can obstruct
representation, whereas span-preserving scrambles primarily break
shared-readout compatibility and sample efficiency.

\subsection{Noise, hardware, and approximate symmetry}
Across the dissipative reservoir's dephasing/shot-noise grid, gains range
from $+4.4\%$ on Kuramoto-Sivashinsky to $+30.0\%$ on advection. The
$K=16$ diagnostics reach $+55.8\%$. At 400 training steps, all 81
combinations of task and noise are robust-positive. On IBM hardware, reuse of the same
raw counts lowers mean NRMSE from 0.084 to 0.057, a $+31.8\%$ gain
($95\%$ CI $[+22.1,+41.9]\%$), with all 10 seeds improving and
$p=q=0.00195$. This is a feasibility check, not a quantum-advantage claim.

On WeatherBench, alignment averages $+5.6\%$ and rises from $+2.8\%$ at
80 training steps to $+11.4\%$ at 10, with all 100 pooled seed gains
positive. The narrow seed-bootstrap interval $[+5.57,+5.65]\%$ conditions
on the fixed 2015 record, rings, split, and preprocessing. It does not
measure uncertainty across years, regimes, latitude samples, or climate.

\subsection{Practical design implications}

The experiments suggest that symmetry should be treated as an end-to-end
design constraint rather than as a property of the quantum dynamics alone.
Making the encoding or reservoir symmetric is insufficient when the measured
features are not organized consistently across sensors or when separate
readouts are trained at each location. In the present architecture, the
largest gains appear only when encoding, dynamics, measurements, and readout
follow the same cyclic convention. Thus, checking symmetry only at the level
of the quantum state can give a misleading picture of the symmetry available
to the trained model.

For practitioners, the measurement/readout interface is the most important
diagnostic point. When shifted copies of relevant observables are not measured,
useful information may be absent from the classical feature space. When the
same information is measured but assigned inconsistently to sensor blocks, an
unrestricted dense readout may still recover it, but the benefit of parameter
sharing is largely lost. A practical workflow is therefore to first verify
that measured features transform consistently with the task symmetry, and
only then impose a shared readout. Observable-orbit completion provides one
systematic way to achieve this, although other symmetry-compatible feature
constructions may serve the same purpose.

The results also indicate when alignment is most likely to help. It is
particularly valuable in data-limited settings, where sharing one local rule
across all sensor positions increases the effective number of training
examples per learned coefficient. The benefit is strongest when the task
closely follows the assumed symmetry, as in advection, and smaller when the
symmetry is only approximate or the dynamics are strongly chaotic. When the
match between assumed symmetry and data is uncertain, aligned and
deliberately misaligned feature maps provide a useful empirical test of
whether the proposed symmetry is genuinely predictive.

Finally, the E/H/O/R construction is one sufficient realization that makes
each source of symmetry breaking experimentally visible. The broader design
principle is that the quantum representation, the measured coordinates, and
the classical hypothesis class should transform compatibly. This principle
is likely to be more transferable than the specific reservoir circuit or
observable library used in the experiments.

\section{Scope and Limitations}
\label{sec:limitations}
What the theory guarantees is deliberately narrow. If the quantum state
shifts whenever the input shifts, if every observable is measured together
with all its shifted copies, and if all sensors share one readout, then
shifted inputs provably lead to shifted forecasts. This is a sufficient
recipe, not one every good design must follow.

The headline $40.1\%$ gain compares the fully aligned model with a
deliberately extreme reference in which nothing is shared across sensors. Both models measure equally many observables, but
the aligned one trains far fewer parameters, because sharing is precisely
the choice under study. Whether alignment helps in realistic settings is
shown by the matched controls, not by this single comparison.

The experiments have clear boundaries. The symmetry is always known in
advance, and every model predicts one step ahead. The hardware run shows
the effect survives on a real device, using a small advection task, eight
sensors, ten reservoir draws, and 256 shots per circuit; it says nothing
about larger sizes, faster runtimes, or other devices. The weather study
uses one variable on five latitude rings over five months; its error bars
capture only reservoir randomness, not variation across years, variables,
or regimes.

Several questions stay open. We do not iterate forecasts over many steps
or learn unknown symmetries from data. Our rings stop at twelve sensors,
sixteen in the noisy simulations. Groups richer than plain shifts may
behave differently, and a symmetry that holds only approximately comes
with no guarantee. Although orbit completion grows the measured list only
linearly, the number of measurement settings needed depends on the device
and the chosen observables, so measuring orbits cheaply remains open.

\section{Conclusion}
\label{sec:conclusion}
In quantum reservoir computing, the trained model sees nothing of the
quantum system except the measured numbers. A task symmetry can therefore
live in the quantum state and still be useless, lost on the way to the
readout. We showed where this loss happens and how to prevent it: measure
every observable together with all its shifted copies, and share one
prediction rule across sensors. The design built this way consistently
outperforms misaligned alternatives in simulation, on IBM hardware, and on
real weather data. The broader lesson is simple. Quantum states,
measurements, and classical learning should be designed together, so that
the structure the physics provides is still there where the learning
actually happens.

\appendices
\section{Reproducibility}
The companion repository at \url{https://github.com/eybmits/qrc-symmetry}
contains everything needed to reproduce each reported number and figure.
It provides complete circuit and encoding definitions, parameter distributions,
observable-channel lists, qubit mappings, transpilation, backend and
calibration metadata, ridge grids, normalization and confidence-interval
procedures, random seeds, WeatherBench latitudes and window origins, and
the raw and derived hardware records.

\balance
\bibliographystyle{IEEEtranDoi}
\bibliography{references}
\end{document}